\documentclass[preprint,12pt]{elsarticle}

\journal{}

\usepackage{aeguill}
\usepackage{amsmath}
\usepackage{amsfonts}
\usepackage{amsthm}
\usepackage{multirow} 
\newcommand\finprovisoire{
\bibliographystyle{elsarticle-num}
\bibliography{/Users/bournez/bibliographie/Bib-Files/bournez,/Users/bournez/bibliographie/Bib-Files/perso}
\end{document}
}

\newtheorem{remark}{Remark}
\newtheorem{definition}{Definition}
\newtheorem{example}{Example}
\newtheorem{theorem}{Theorem}
\newtheorem{proposition}{Proposition}
\newtheorem{corollary}{Corollary}

\newcommand\vx{{x}}

\newcommand\vy{{y}}
\newcommand\vz{{z}}
\newcommand\R{\mathbb{R}}
\newcommand\motnouv[1]{\emph{#1}}

\title{Population Protocols that Correspond to Symmetric Games\tnoteref{label1}}
 \tnotetext[label1]{This work and all authors were partly supported by ANR Project SOGEA and by ANR Project SHAMAN, Xavier Koegler was partly supported by COST Action 295 DYNAMO and ANR Project ALADDIN}

\author[lix]{Olivier Bournez}\ead{Olivier.Bournez@lix.polytechnique.fr}
\author[mrs]{J\'er\'emie Chalopin}\ead{Jeremie.Chalopin@lif.univ-mrs.fr}
\author[prism]{Johanne Cohen}\ead{Johanne.Cohen@prism.uvsq.fr}
\author[liafa]{Xavier Koegler}\ead{Xavier.Koegler@liafa.jussieu.fr}

\address[lix]{Ecole Polytechnique \& Laboratoire d'Informatique (LIX),\\ 91128 Palaiseau Cedex, France}
\address[mrs]{CNRS \& Laboratoire d'Informatique Fondamentale de Marseille, CNRS \& Aix-Marseille Universit{\'e},\\ 39 rue Joliot Curie, 13453 Marseille Cedex 13, France}
\address[prism]{{CNRS \& PRiSM},\\ 45 Avenue des Etats Unis, 78000 Versailles, France}
\address[liafa]{\'{E}cole Normale Supérieure \& Université Paris Diderot - Paris 7,\\ Case 7014, 75205 Paris Cedex 13, France}

\begin{document}

\begin{frontmatter}


\begin{abstract}
  Population protocols have been introduced by Angluin et {al.} as a model of 
 networks consisting of very limited mobile agents that interact in pairs but with no control
 over their own movement. A collection of anonymous agents, modeled
 by finite automata, interact pairwise according to some rules that update their states.

 The model has been considered as a computational model in several papers. Input values are initially distributed among the agents, and the agents must eventually converge to the the correct output.  Predicates on the initial configurations that
  can be computed by such protocols have been characterized under
  various hypotheses. The model has initially been motivated by sensor-networks, but it  can be seen more generally as a model of networks of anonymous agents interacting pairwise. This includes sensor networks, ad-hoc networks, or models from chemistry.

  In an orthogonal way, several distributed systems have been termed in literature as being realizations of games  in the sense of game theory. In this paper, we investigate under which conditions population protocols, or more generally pairwise interaction rules, can be considered as the result of a symmetric game. We prove that not all rules can be considered as symmetric games.
  We characterize the computational power of symmetric games. We prove that they have very limited power: they can count until $2$, compute majority, but they can not even count until $3$.

  As a side effect of our study, we also prove that any population protocol can be simulated by a symmetric one (but not necessarily a game).  \end{abstract}


\end{frontmatter}

\section{Introduction}

The computational power of networks of anonymous resource-limited
mobile agents has been investigated recently. 

In
particular, Angluin et al.\ proposed in \cite{AspnesADFP2004} a model of distributed 
computations. In this model, called \motnouv{population
protocols}, finitely many finite-state agents interact in pairs
chosen by an adversary. Each interaction has the effect of updating
the state of the two agents according to a joint transition function.

A protocol is said to \motnouv{(stably) compute} a predicate on the initial states
of the agents if, in any fair execution, after finitely many
interactions, all agents reach a common output that corresponds to the
value of the predicate.

The model was originally proposed to model computations realized by
sensor networks in which passive agents are carried along by other
entities. The canonical example of \cite{AspnesADFP2004} corresponds to
sensors attached to a flock of birds and that must be programmed to
check some global properties, like determining whether more than 5\%
of the population has elevated temperature. Motivating scenarios also
include models of the propagation of trust 
\cite{diamadi2001sgs}.

Much of the work so far on population protocols has concentrated on
characterizing which predicates on the initial states can be computed
in different variants of the model and under various assumptions. In
particular, the  predicates computable by the unrestricted 
population protocols from \cite{AspnesADFP2004} have been characterized
as being precisely the semi-linear predicates, that is those
predicates on counts of input agents definable in first-order
Presburger arithmetic \cite{presburger:uvk}. Semi-linearity was shown
to be sufficient in \cite{AspnesADFP2004} and necessary in
\cite{AngluinAE2006semilinear}.

Variants considered so far include restriction to one-way
communications, restriction to particular interaction graphs, to
random interactions, with possibly various kind of failures of
agents. Solutions to classical problems of distributed algorithmics
have also been considered in this model. Refer to 
\cite{PopProtocolsEATCS} for a survey and a complete discussion.

The population protocol model shares many features with other models already considered in the literature. In particular, models of pairwise interactions have been used to study the propagation of diseases \cite{Heth00}, or rumors \cite{dk65}. In chemistry, the chemical master equation has been justified using (stochastic) pairwise interactions between the finitely many molecules  \cite{Murray-VolI,gillespie1992rdc}. In that sense, the model of population protocols may be considered as fundamental in several fields of study, since it appears as soon as anonymous agents interact pairwise.
 
In an orthogonal way, pairwise interactions between finite-state agents are sometimes motivated by the study of the dynamics of particular two-player games from game theory. For example, paper \cite{Ref9deFMP04} considers the dynamics of the so-called $PAVLOV$ behavior in the iterated prisoner lemma. Several results about the time of convergence of this particular dynamics towards the stable state can be found in \cite{Ref9deFMP04}, and \cite{FMP04}, for rings, and complete graphs.

The purpose of this article is to better understand whether and when pairwise interactions, and hence population protocols, can be considered as the result of a game.
We want to understand if restricting to rules that come from a symmetric game is a limitation, and in particular whether restricting to rules that can be termed $PAVLOV$ in the spirit of \cite{Ref9deFMP04} is a limitation.  We do so by giving solutions to several basic problems using rules of interactions associated to a symmetric game, and by characterizing they power: We prove that they can count until $2$, they can compute $MAJORITY$, but they can not even count until $3$.

As such protocols must also be symmetric, we are also discussing whether restricting to symmetric rules in population protocols is a limitation. We prove that any population protocol can be simulated by a symmetric one (but not necessarily a game).
 
In Section \ref{section:pp}, we briefly recall population protocols. In Section \ref{section:gametheory}, we recall some basics from game theory. In Section \ref{sec:gamepp}, we discuss how a game can be turned into a dynamics, and introduce the notion of {Pavlovian} population protocol. In Section \ref{sec:results} we prove that any symmetric deterministic 2-states population protocol is Pavlovian, and that the problem of computing the OR, AND, as well as the leader election and majority problem admit Pavlovian solutions. We then characterize there power by proving that they can count until $2$, but they can not count until $3$ in Section   \ref{sec:bound}. We prove that symmetric population protocols, unlike the restricted class of Pavlovian population protocols can compute all semi-linear predicate in Section \ref{sec:fin}.

\paragraph{Related work}
Population protocols have been introduced in \cite{AspnesADFP2004}, and proved to compute all semi-linear predicates. They have been proved not to be able to compute more in \cite{AngluinAE2006semilinear}. Various restrictions on the initial model have been considered up to now. 
An (almost) up to date survey can be found in \cite{PopProtocolsEATCS}. 

Variants include discussions about the influence of removing the assumption of two-way interaction: One-way interaction models include variants where agents communicate by anonymous message-passing, with immediate delivery or delayed delivery, or where agents can record it has sent a message, or queue incoming messages \cite{angluin2007cpp}. However, as far as we know, the constraint of restricting to symmetric rules has not been yet explicitly considered, nor restricting to rules that correspond to games in the population protocol literature.

More generally, population protocols arise as soon as populations of anonymous agents interact in pairs. Our original motivation was to consider rules corresponding to two-players games, and population protocols arose quite incidentally. The main advantage of the  \cite{AspnesADFP2004} settings is that it provides a clear understanding of what is called a computation by the model. Many distributed systems have been described as the result of games, but as far as we know there has not been attempts to characterize what can be computed by games in the spirit of this computational model.

In this paper, we turn two players games into dynamics over agents, by considering $PAVLOV$ behavior. This is inspired by \cite{Ref9deFMP04,FMP04,kraines1988psd} that consider the dynamics of a particular set of rules termed the $PAVLOV$ behavior in the iterated prisoner lemma. The $PAVLOV$ behavior is sometimes also termed \textit{WIN-STAY, LOSE-SHIFT} \cite{nowak1993sws,axelrod:1984:ec}. Notice, that we extended it from two-strategies two-players games to n-strategies two-players games, whereas above references only talk about two-strategies two-players games, and mostly of the iterated prisoner lemma.

This is clearly not the only way to associate a dynamic to a game. They are several famous classical approaches: The first consists in repeating games: see for example \cite{LivreGameTheory,LivreBinmore}. The second in using models from evolutionary game theory: refer to \cite{Evolutionary1,LivreWeibull} for a presentation of this latter approach. The approach considered here falls in method that consider dynamics obtained by selecting at each step some players and let them play a fixed game.  Alternatives to $PAVLOV$ behavior could include $MYOPIC$ dynamics (at each step each player chooses the best response to previously played strategy by its adversary), or the well-known and studied $FICTIOUS-PLAYER$ dynamics (at each step each player chooses the best response to the statistics of the past history of strategies played by its adversary). We refer to \cite{theorylearninggames,LivreBinmore} for a presentation of results known about the properties of the obtained dynamics according to the properties of the underlying game. This is clearly non-exhaustive, and we refer to \cite{axelrod:1984:ec} for an incredible zoology of possible behaviors for the particular iterated prisoner lemma game, with discussions of their compared merits in experimental tournaments.

We obtain a characterization of the power of Pavlovian population protocols in terms of closure properties that show that they can count until $2$, but not until $3$. Notice that several variants of (one-way) population protocols have been characterized in \cite{angluin2007cpp} in a $COUNT_k$ hierarchy. The class obtained here seems close to the $COUNT_2$ level of this latter hierarchy \cite{angluin2007cpp}. However, on one hand, this is not exactly this class (for example $MAJORITY$ is computable but not in the $COUNT_2$ level), and on the other hand, as no class there is formally proved to correspond to $COUNT_2$, this shows that the class considered here is different, and not reducible to the variants of \cite{angluin2007cpp}.

Notice that a preliminary version of this article has been presented in \textit{Complexity of Simple Programs CSP'08}. Compared to this conference version, we simplified some constructions, we added a few protocols, we extended deeply related work discussions, and mainly, we solved the statements conjectured there: we provide here a characterization of the power of Pavlovian population protocols, whereas it was open at the time of the presentation of this preliminary version.

\section{Population Protocols}
\label{section:pp}

A protocol \cite{AspnesADFP2004,PopProtocolsEATCS} is given by $(Q,\Sigma,\iota,\omega,\delta)$ with the
following components. $Q$ is a finite set of \motnouv{states}.
$\Sigma$ is a finite set of \motnouv{input symbols}.  $\iota: \Sigma
\to Q$ is the initial state mapping, and $\omega: Q \to \{0,1\}$ is
the individual output function. $\delta \subseteq Q^4$ is a joint
transition relation that describes how pairs of agents can
interact. Relation $\delta$ is sometimes described by listing all
possible interactions using the notation $(q_1,q_2) \to (q'_1,q'_2)$,
or even the notation $q_1q_2 \to q'_1 q'_2$, 
for $(q_1,q_2,q'_1,q'_2) \in \delta$ (with the convention that
$(q_1,q_2) \to (q_1,q_2)$ when no rule is specified with $(q_1,q_2)$
in the left-hand side). The protocol is termed \motnouv{deterministic}
if for all pairs $(q_1,q_2)$ there is only one pair $(q'_1,q'_2)$ with
$(q_1,q_2) \to (q'_1,q'_2)$. In that case, we write
$\delta_1(q_1,q_2)$ for the unique $q'_1$ and $\delta_2(q_1,q_2)$ for
the unique $q'_2$.   

Notice that, in general, rules can be non-symmetric: if
$(q_1,q_2) \to (q'_1,q'_2)$, it does not necessarily follow that $(q_2,q_1) \to
(q'_2,q'_1)$.  

Computations of a protocol proceed in the following way. The
computation takes place among $n$ \motnouv{agents}, where $n \ge 2$. A
\motnouv{configuration} of the system can be described by a vector of
all the agents' states. The state of each agent is an element of $Q$. Because agents
with the same states are indistinguishable, each configuration can be
summarized as an unordered multiset of states, and hence of elements
of $Q$.

Each agent is given initially some input value from $\Sigma$: Each agent's initial
state is determined by applying $\iota$ to its input value. This
determines the initial configuration of the population.

An execution of a protocol proceeds from the initial configuration by
interactions between pairs of agents. Suppose that two agents in state
$q_1$ and $q_2$ meet and have an interaction. They can change into
state $q'_1$ and $q'_2$ if $(q_1,q_2,q'_1,q'_2)$ is in the transition
relation $\delta$.  If $C$ and $C'$ are two configurations, we write
$C \to C'$ if $C'$ can be obtained from $C$ by a single interaction of
two agents: this means that $C$ contains two states $q_1$ and $q_2$ and $C'$ is
obtained by replacing $q_1$ and $q_2$ by $q'_1$ and $q'_2$ in $C$,
where $(q_1,q_2,q'_1,q'_2) \in \delta$. An \motnouv{execution} of the
protocol is an infinite sequence of configurations
$C_0,C_1,C_2,\cdots$, where $C_0$ is an initial configuration and $C_i
\to C_{i+1}$ for all $i\ge0$. An execution is \motnouv{fair} if for
all configurations $C$ that appear infinitely often in the execution,
if $C \to C'$ for some configuration $C'$, then $C'$ appears
infinitely often in the execution.

At any point during an execution, each agent's state determines its
output at that time. If the agent is in state $q$, its output value is
$\omega(q)$. The configuration output is $0$ (respectively $1$) if all
the individual outputs are $0$ (respectively $1$). If the individual
outputs are mixed $0$s and $1s$ then the output of the configuration
is undefined. 

Let $p$ be a predicate over multisets of elements of
$\Sigma$. Predicate $p$ can be considered as a function whose range is
$\{0,1\}$ and whose domain is the collection of these multisets. The predicate is said to be computed by the protocol if,  for every  multiset $I$, and
every fair execution that starts from the initial configuration
corresponding to $I$, the output value of every agent eventually
stabilizes to $p(I)$.

Multisets of elements of $\Sigma$ are in clear bijection with elements of $\mathbb{N}^{|\Sigma|}$: a multiset over $\Sigma$ can be identified by a vector of $|\Sigma|$ components, where each component represents the multiplicity of the corresponding element of $\Sigma$ in this multiset. It follows that predicates can also be considered as  functions whose range is $\{0,1\}$ and whose domain is $\mathbb{N}^{|\Sigma|}$. 

The following was then proved in
\cite{AspnesADFP2004,AngluinAE2006semilinear}.

\begin{theorem}[\cite{AspnesADFP2004,AngluinAE2006semilinear}] A
  predicate is computable in the population protocol model if and only
  if it is semilinear.
\end{theorem}

Recall that semilinear sets are known to correspond to predicates on
counts of input agents definable in first-order Presburger arithmetic
\cite{presburger:uvk}.

We will use the following notation as in \cite{angluin2007cpp}: the set of all functions from a set $X$ to a set $Y$ is denoted by $Y^X$. Let $\Sigma$ be a finite non-empty set. For all $f,g \in \R^ E$, we define the usual vector space operations 
$$\begin{array}{llll}
(f+g)(\sigma) &=& f(\sigma) + g(\sigma)& \mbox{ for all }  \sigma \in \Sigma \\
(f-g)(\sigma) &=& f(\sigma) - g(\sigma)& \mbox{ for all }  \sigma \in \Sigma \\
(cf)(\sigma) &=& c f(\sigma) & \mbox{ for all }  \sigma \in \Sigma , c \in \R\\
(f.g)(\sigma) &=& \sum_{\sigma} f(\sigma) g(\sigma).& \\
\end{array}
$$

Abusing notation as in \cite{angluin2007cpp}, we will write $\sigma$ for the function $\sigma(\sigma')=[\sigma=\sigma']$, for all $\sigma' \in \Sigma$, where $[condition]$ is $1$ if condition is true, $0$ otherwise.

\section{Game Theory}
\label{section:gametheory}

We now recall the simplest concepts from Game Theory. We focus on non-cooperative games, with complete information, in extensive form. 

The simplest game is made up of two players, called $I$ and $II$, with a
finite set of options, called \emph{pure strategies}, $Strat(I)$ and
$Strat(II)$. Denote by $A_{i,j}$ (respectively: $B_{i,j}$) the score
for player $I$
(resp. $II$) when $I$ uses strategy $i \in Strat(I)$ and $II$ uses strategy
$j \in Strat(II)$.
 
The scores are given by $n \times m$ matrices $A$ and $B$, where $n$ and
$m$ are the cardinality of $Strat(I)$ and $Strat(II)$. The game is
termed \emph{symmetric} if $A$ is the transpose of $B$: this implies
that $n=m$, and we can assume without loss of generality that
$Strat(I)=Strat(II)$. 

In this paper, we will restrict to symmetric games.

\begin{example}[Prisoner's dilemma]
The case where $A$ and $B$ are the following matrices 

$$A = \left(
\begin{array}{ll}
R & S \\
T & P \\
\end{array}
\right)
, 
B= 
\left(
\begin{array}{ll}
R & T \\
S & P \\
\end{array}
\right)
$$
with $T>R>P>S$ and $2R >T+S$, is called the \emph{prisoner's dilemma}. We denote by $C$ (for
cooperation) the first pure strategy, and by $D$ (for defection) the
second pure strategy of each player.

As the game is symmetric, matrix $A$ and $B$ can also be denoted by:

\begin{center}
\begin{tabular}{llcr} 
  &   & \multicolumn{2}{c}{Opponent}  \\
  &  &  { C} &{ D } \\ \cline{3-4}
\multirow{2}{*}{Player} & C &  \multicolumn{1}{|c}{$R$} &\multicolumn{1}{r|}{$S$}\\
                        & D &  \multicolumn{1}{|c}{$T$} &\multicolumn{1}{r|}{$P$}\\ \cline{3-4}
\end{tabular}
\end{center}
\end{example}





A strategy $x \in Strat(I)$ is said to be a best response to strategy
$y \in Strat(II)$, denoted by $x \in BR(y)$ if
\begin{equation}
A_{z,y} \le A_{x,y}
\end{equation}
for all strategies $z \in Strat(I)$.


A pair $(x,y)$ is a \emph{(pure) Nash equilibrium} if $x \in
BR(y)$ and $y \in BR(x)$. A pure Nash equilibrium does not always
exist. 


In other words, two strategies $(x,y)$ form a Nash equilibrium if
in that state neither of the players has a unilateral interest to deviate
from it.


\begin{example}
  On the example of the prisoner's dilemma, $BR(\vy)=D$ for all
  $\vy$, and $BR(\vx)=D$ for all $\vx$. So $(D,D)$ is the unique
  Nash equilibrium, and it is pure. In it, each player has score
  $P$. The well-known paradox is that if they had played $(C,C)$ (cooperation)
  they would have had score $R$, that is more. The social optimum
  $(C,C)$, is different from the equilibrium that is reached by
  rational players $(D,D)$, since in any other state, each player
  fears that the adversary plays $C$.
\end{example}

We will also introduce the following definition: Given some strategy $\vx' \in Strat(I)$, a strategy $\vx \in Strat(I)$ is said to be a best response to strategy
$\vy \in Strat(II)$ among those different from $\vx'$, denoted by $\vx
\in BR_{\neq \vx'}(\vy)$ if
\begin{equation}
A_{z,y} \le A_{x,y}
\end{equation}
for all strategy $\vz \in Strat(I), \vz \neq \vx'$. 

Of course, the roles of $II$ and $I$ can be inverted in the previous definition.


There are two main approaches to discuss dynamics of games. The first
consists in repeating games \cite{LivreGameTheory,LivreBinmore}. The second in using models from
evolutionary game theory. Refer to \cite{Evolutionary1,LivreWeibull}
for a presentation of this latter approach.


\paragraph{Repeating Games}

Repeating $k$ times a game, is equivalent to extending the space of
choices into $Strat(I)^k$ and $Strat(II)^k$: player $I$ (respectively
$II$) chooses his or her action $\vx(t) \in Strat(I)$, (resp. $\vy(t) \in
Strat(II)$) at time $t$ for $t=1,2,\cdots,k$.  Hence, this is
equivalent to a two-player game with respectively $n^k$ and $m^k$
choices for players.

To avoid confusion, we will call
\emph{actions} the choices $\vx(t),\vy(t)$ of each player at a given time,
and \emph{strategies} the sequences $X=\vx(1),\cdots,\vx(k)$ and
$Y=\vy(1),\cdots,\vy(k)$, that is to say the strategies for the global
game.

If the game is repeated an infinite number of times, a strategy
becomes a function from integers to the set of actions, and the game
is still equivalent to a two-player game\footnote{but whose matrices
  are infinite.}.

\paragraph{Behaviors}

In practice, player $I$ (respectively $II$) has to solve the following
problem at each time $t$: given the history of the game up to now,
that is to say $$X_{t-1}=\vx(1),\cdots,\vx(t-1)$$ and
$$Y_{t-1}=\vy(1),\cdots,\vy(t-1)$$ what should I (resp. II) play at time $t$? In
other words, how to choose $\vx(t) \in Strat(I)$?  (resp. $\vy(t) \in
Strat(II)$?)

Is is natural to suppose that this is
given by some behavior rules: $$\vx(t)=f(X_{t-1},Y_{t-1}),$$
$$\vy(t)=g(X_{t-1},Y_{t-1})$$ for some particular functions $f$ and $g$.

\paragraph{The Specific Case of the Prisoner's Lemma} 
The question of the best behavior rule to use for the prisoner lemma
gave birth to an important literature. In particular, after the book
\cite{axelrod:1984:ec}, that describes the results of tournaments of
behavior rules for the iterated prisoner lemma, and that argues that
there exists a best behavior rule called $TIT-FOR-TAT$.
This  consists in cooperating at the first step,
and then do the same thing as the adversary at subsequent times. %

A lot of other behaviors, most of them with very
picturesque names have been proposed and studied: see for example
\cite{axelrod:1984:ec}.

Among possible behaviors there is $PAVLOV$ behavior: in the iterated
prisoner lemma, a player cooperates if and only if both players opted
for the same alternative in the previous move. This name 
\cite{axelrod:1984:ec,kraines1988psd,nowak1993sws} stems from the fact that this strategy embodies
an almost reflex-like response to the payoff: it repeats its former
move if it was rewarded by $R$ or $T$ points, but switches behavior if
it was punished by receiving only $P$ or $S$ points. Refer to
\cite{nowak1993sws} for some study of this strategy in the spirit of
Axelrod's tournaments.

The $PAVLOV$ behavior can also be termed \textit{WIN-STAY, LOSE-SHIFT} since
if the play on the previous round results in a success, then the
agent plays the same strategy on the next round. Alternatively, if the
play resulted in a failure the agent switches to another action
\cite{axelrod:1984:ec,nowak1993sws}.

\paragraph{Going From $2$ Players to $N$ Players} $PAVLOV$ behavior is Markovian: a behavior $f$ is \emph{Markovian}, if
$f(X_{t-1},Y_{t-1})$ depends only on $\vx(t-1)$ and $\vy(t-1)$.

From such a behavior, it is easy to obtain a distributed
dynamic. For example, let's follow \cite{Ref9deFMP04}, for the
prisoner's dilemma.

Suppose that we have a connected graph $G=(V,E)$, with $N$
vertices. The vertices correspond to players. An instantaneous
configuration of the system is given by an element of $\{C,D\}^N$,
that is to say by the state $C$ or $D$ of each vertex. Hence, there
are $2^N$ configurations.

At each time $t$, one chooses randomly and uniformly one edge $(i,j)$
of the graph. At this moment, players $i$ and $j$ play the prisoner
dilemma with the $PAVLOV$ behavior. It is easy to see that this
corresponds to executing the following rules:

\begin{equation} \label{eq:pavlov}
\left\{
\begin{array}{lll}
CC &\to& CC \\
CD &\to& DD \\
DC &\to& DD \\
DD &\to& CC. \\
\end{array}
\right.
\end{equation}

What is the final state reached by the system?  The underlying model is a very large Markov chain with $2^N$ states. The
state $E^*=\{C\}^N$ is absorbing. If the graph $G$ does not have any
isolated vertex, this is the unique absorbing state, and there exists
a sequence of transformations that transforms any state $E$ into this
state $E^*$.
As a consequence, from well-known classical results in Markov chain theory, whatever  the initial configuration is, with
probability $1$, the system will eventually be in state $E^*$ 
\cite{Markov}. The system is \emph{self-stabilizing}.

Several results about the time of convergence  towards this stable state
can be found in \cite{Ref9deFMP04}, and \cite{FMP04}, for  rings,
and  complete graphs.

What is interesting in this example is that it shows how to go from a 
game, and a behavior to a distributed dynamic on a
graph, and in particular to a population protocol when the graph is
a complete graph.  

\section{From Games To Population Protocols}
\label{sec:gamepp}

In the spirit of the previous discussion, to any symmetric game, we can
associate a population protocol as follows.

\begin{definition}[Associating a Protocol to a Game]
  Assume a symmetric two-player game is given. Let $\Delta$ be
  some threshold. 
  
  The  protocol associated to the game  is a population
  protocol whose set of states is $Q$, where  $Q=Strat(I)=Strat(II)$
  is the set of strategies of the game, and whose transition rules
  $\delta$ are
  given as follows:
$$(q_1,q_2,q'_1,q'_2) \in \delta$$ where
  \begin{itemize}
  \item $q'_1=q_1$ when $M_{q_1,q_2}\ge\Delta$
  \item $q'_1 \in BR_{\neq q_1}(q_2)$ when $M_{q_1,q_2} < \Delta$
  \end{itemize}
and
   \begin{itemize}
  \item $q'_2=q_2$ when $M_{q_2,q_1} \ge \Delta$
  \item $q'_2 \in BR_{\neq q_2}(q_1)$ when $M_{q_2,q_1} < \Delta$,
  \end{itemize}
where $M$ is the matrix of the game.
\end{definition}

\begin{remark} By subtracting $\Delta$ to each entry of the matrix $M$, we can assume without loss of generality that $\Delta=0$. We will do so from now on.
\end{remark}

\begin{definition}[Pavlovian Population Protocol]
A population protocol is \motnouv{Pavlovian} if it can be obtained
from a game as above.
\end{definition}

\begin{remark}
Clearly a Pavlovian population protocol must be \motnouv{symmetric}:
indeed, whenever
$(q_1,q_2,q'_1,q'_2) \in \delta$, one has $(q_2,q_1,q'_2,q'_1) \in
\delta$.
\end{remark}

\section{Some Specific Pavlovian Protocols} 
\label{sec:results}


We now discuss whether assuming protocols Pavlovian is a restriction.

We start by an easy consideration.

\begin{theorem} Any symmetric deterministic $2$-states population protocol is Pavlovian.
\end{theorem}

\begin{proof}
  Consider a deterministic symmetric $2$-states population protocol. Note $Q=\{+,-\}$ its set
  of states. Its transition function can be written as follows:
 
\begin{equation}
\left\{
\begin{array}{lll}
++ & \to & \alpha_{++} \alpha_{++}  \\
+- & \to &  \alpha_{+-} \alpha_{-+}\\
-+ & \to & \alpha_{-+} \alpha_{+-}\\
-- & \to & \alpha_{--}  \alpha_{--}\\
\end{array}
\right.\ 
\end{equation}
for some  $\alpha_{++},\alpha_{+-},\alpha_{-+},\alpha_{--}$.

This corresponds to the symmetric game given by the following pay-off matrix $M$  
\begin{center}
\begin{tabular}{llcc} 
  &   & \multicolumn{2}{c}{Opponent}  \\
  &  &  \multicolumn{1}{c}{ \textsc{+}}     & { \textsc{-} }  \rule[-7pt]{0pt}{20pt}\\ \cline{3-4}
 \multirow{2}{*}{Player} & \textsc{+} &  \multicolumn{1}{|c}{$\beta_{++}$}  &\multicolumn{1}{c|}{$\beta_{+-}$}\rule[-7pt]{0pt}{20pt}\\
                        & \textsc{-}  &  \multicolumn{1}{|c}{$\beta_{-+}$}  &\multicolumn{1}{c|}{$\beta_{--}$}\rule[-7pt]{0pt}{20pt}\\ \cline{3-4}
\end{tabular}
\end{center}
where for all $q_1,q_2 \in \{+,-\}$,
\begin{itemize}
\item  $\beta_{q_1q_2}=1$ if  $\alpha_{q_1 q_2}=q_1$,
\item $\beta_{q_1q_2}=-1$  otherwise.
\end{itemize}


 
\end{proof}

Unfortunately, not all rules correspond to a game.

\begin{proposition}
Some symmetric population protocols are not Pavlovian.
\end{proposition}

\begin{proof}
Consider for example a deterministic $3$-states population protocol
with set of states $Q=\{q_0,q_1,q_2\}$ and a joint transition function
$\delta$ such that $\delta_1(q_0,q_0)=q_{1}$,
$\delta_1(q_1,q_0)=q_{2}$ , $\delta_1(q_2,q_0)=q_{0}$.

Assume by contradiction that there exists a $2$-player game
corresponding to this $3$-states population protocol. Consider its
payoff matrix $M$. Let $M(q_0,q_0)=\beta_{0}$, $M(q_1,q_0)=\beta_{1}$
, $M(q_2,q_0)=\beta_{2}$. We must have $\beta_0 \ge \Delta=0, \beta_1
\ge \Delta=0$ since all agents that interact with
an agent in state $q_0$ must change their state. Now, since $q_0$
changes to $q_1$, $q_1$ must be a strictly better response to $q_0$ than
$q_2$: hence, we must have $\beta_1 > \beta_2$.  In a similar way,
since $q_1$ changes to $q_2$, we must have $\beta_2 > \beta_0$ , and
since $q_2$ changes to $q_0$, we must have $\beta_0 > \beta_1$. From
$\beta_1 > \beta_2 > \beta_0$ we reach a contradiction.
\end{proof}

This indeed motivates the following study, where we discuss which problems
admit a Pavlovian solution.

\subsection{Basic Protocols}

\begin{proposition} There is a Pavlovian protocol that computes the
  logical $OR$ (resp. $AND$) of input bits.
\end{proposition}

\begin{proof} 
Consider the following protocol to compute $OR$, 

\begin{equation}
\left\{
\begin{array}{lll}
01 & \to & 11 \\
10 & \to & 11  \\
00 & \to & 00 \\
11 & \to & 11\\
 \end{array}
\right.
\end{equation}

and the following protocol to compute $AND$, 

\begin{equation}
\left\{
\begin{array}{lll}
01 & \to & 00 \\
10 & \to & 00  \\
00 & \to & 00 \\
11 & \to & 11\\
 \end{array}
\right.
\end{equation}

Since they are both deterministic 2-states population protocols,
they are
Pavlovian. 

\end{proof}

\begin{remark}
Notice that $OR$  (respectively $AND$) protocol corresponds to the predicate $[x.0 \ge 1]$ (resp. $[x.0 =0]$), where $x$ is the input. A simple change of notation yields a  protocol to compute $[x.\sigma \ge 1]$ and $[x.\sigma =0]$ for any input symbol $\sigma$.
\end{remark}

\begin{remark}
All previous protocols are ``naturally broadcasting'' i.e., eventually all agents
agree on some (the correct) value. With previous definitions (which are the
classical ones for population protocols), the  following protocol
does not compute the $XOR$ or input bits, or equivalently does not
compute predicate $[x.1 \equiv
1~ (mod~2)]$.  \begin{equation}
\left\{
\begin{array}{lll}
01 & \to & 01 \\
10 & \to & 10  \\
00 & \to & 00 \\
11 & \to & 00\\
 \end{array}
\right.
\end{equation}

Indeed, the answer is not eventually known
by all the agents. It computes the $XOR$ in a weaker form i.e., eventually, all
agents will be in state $0$, if the $XOR$ of input bits is $0$, or
eventually only one agent will be in state $1$, if the $XOR$ of input
bits is $1$.
\end{remark}

\begin{proposition} There is a Pavlovian protocol that computes the threshold predicate $[x.\sigma \ge 2]$, which is true when there are at least $2$ occurrences of input symbol $\sigma$ in the input $x$.
\end{proposition}

\begin{proof} 
The following protocol is a solution taking
\begin{itemize}\item $\Sigma=\{0,\sigma\}$,
  $Q=\{0,\sigma,2\}$, 
\item $\omega(0)=\omega(\sigma)=0$, 
\item $\omega(2)=1$. \end{itemize}

\begin{equation} \label{ode:dynamic}
\left\{
\begin{array}{lll}
00 \to & 00\\
0\sigma \to & 0\sigma\\
\sigma0 \to & \sigma0\\
02 \to & 22\\
20 \to & 22\\
\sigma\sigma \to & 22\\
\sigma2 \to & 22\\
2\sigma \to & 22\\
22 \to & 22\\
 \end{array}
\right.
\end{equation}
Indeed, if there is at least two $\sigma$, then by fairness and by the rule number $6$, they will ultimately be changed into two $2$s. Then $2$s will turn all other agents into $2$s. Now, this is the only way to create a $2$.

This is a Pavlovian protocol since it corresponds to the following payoff matrix.

\begin{center}
\begin{tabular}{llccc} 
  &   & \multicolumn{3}{c}{Opponent}  \\
  &  &  { \textsc{$0$}}     & { \textsc{$\sigma$}  } & {\textsc{$2$}}  
\rule[-7pt]{0pt}{20pt}\\ \cline{3-5}
 \multirow{3}{*}{Player} & \textsc{$0$} &  \multicolumn{1}{|c}{$0$}
 & \multicolumn{1}{c}{$0$}
 &\multicolumn{1}{c|}{$-1$}\rule[-7pt]{0pt}{20pt}\\

 & \textsc{$\sigma$} &  \multicolumn{1}{|c}{$0$}
 & \multicolumn{1}{c}{$-1$}
 &\multicolumn{1}{c|}{$-1$}\rule[-7pt]{0pt}{20pt}\\

 & \textsc{$2$} &  \multicolumn{1}{|c}{$1$}
 & \multicolumn{1}{c}{$1$}
 &\multicolumn{1}{c|}{$1$}\rule[-7pt]{0pt}{20pt}\\
\cline{3-5}

\end{tabular}
\end{center}

\end{proof}

Hence, Pavlovian population protocols can count until $2$. We will prove later on that they can not count until $3$.

\subsection{Leader Election}

The classical solution \cite{AspnesADFP2004} to the leader election problem (starting
from a configuration with $\ge 1$ leaders, eventually exactly one leader
survives) is the following:

\begin{equation} 
\left\{
\begin{array}{lll}
LL & \to & LN \\
LN & \to & LN  \\
NL & \to & NL \\
NN & \to & NN\\
 \end{array}
\right.
\end{equation}

Notice that we use the terminology ``leader election'' as in \cite{AspnesADFP2004} for this protocol, but that it may be considered more as a ``mutual exclusion'' protocol.

Unfortunately, this protocol is  non-symmetric, and hence
non-Pavlovian.

\begin{remark} Actually, the problem is
with the first rule, since one wants two leaders to become only
one. If the two leaders are identical, this is clearly problematic with symmetric rules.
\end{remark}

However, the leader election problem can actually be solved by a Pavlovian
protocol, at the price of a less trivial protocol.




\begin{proposition}
The following Pavlovian protocol solves the leader election problem,
as soon as the population is of size $\ge 3$.
\begin{equation} \label{ode:dynamicd}
\left\{
\begin{array}{lll}
L_1L_2 & \to & L_1N \\
L_1N& \to & NL_2\\
L_2N & \to & NL_1 \\
NN & \to & NN\\
L_2L_1 & \to & N L_1 \\
N L_1 & \to &  L_2 N \\
N L_2 & \to &  L_1  N\\
L_1 L_1 & \to &L_2 L_2 \\
L_2 L_2 & \to &L_1 L_1 \\
 \end{array}
\right.
\end{equation}
\end{proposition}

\begin{proof}
Indeed, starting from a configuration containing not only $N$s,
eventually after some time configurations will have exactly one
leader, that is one agent in state $L_1$ or $L_2$.

Indeed, the first rule and the fifth rule decrease strictly the
number of leaders whenever there are more than two leaders. Now the
other rules, preserve the number of leaders, and are made such that an $L_1$ can always be transformed
into an $L_2$ and vice-versa, and hence are made such that a configuration where
first or fifth rule applies can always be reached whenever there are
more than two leaders. The fact that it solves the leader election
problem then 
follows from the hypothesis of fairness in the definition of
computations.

This is a Pavlovian protocol, since it corresponds to the following
payoff matrix.
 
\begin{center}
\begin{tabular}{llccc} 
  &   & \multicolumn{3}{c}{Opponent}  \\
  &  &  { \textsc{$L_1$}}     & { \textsc{$L_2$}  } & {\textsc{N}}  
\rule[-7pt]{0pt}{20pt}\\ \cline{3-5}
 \multirow{3}{*}{Player} & \textsc{$L_1$} &  \multicolumn{1}{|c}{$-3$}
 & \multicolumn{1}{c}{$0$}
 &\multicolumn{1}{c|}{$-3$}\rule[-7pt]{0pt}{20pt}\\

 & \textsc{$L_2$} &  \multicolumn{1}{|c}{$-1$}
 & \multicolumn{1}{c}{$-3$}
 &\multicolumn{1}{c|}{$-3$}\rule[-7pt]{0pt}{20pt}\\

 & \textsc{$N$} &  \multicolumn{1}{|c}{$-2$}
 & \multicolumn{1}{c}{$-3$}
 &\multicolumn{1}{c|}{$0$}\rule[-7pt]{0pt}{20pt}\\
\cline{3-5}

\end{tabular}
\end{center}
\end{proof}

\subsection{Majority}

\begin{proposition}
The majority problem (given some population of input symbols $\sigma$ and $\sigma'$, determine
whether there are more $\sigma$ than $\sigma'$) can be solved by a Pavlovian population protocol.
\end{proposition}

\begin{remark} If one prefers, the predicate $[x.\sigma \ge x.\sigma']$, where $\sigma$ and $\sigma'$ are two input symbols, and $x$ is the input, is computable by a Pavlovian
population protocol.
\end{remark}

\begin{proof}

  We claim that the following  protocol outputs $1$ if there
  are more $\sigma$ than $\sigma'$ in the initial configuration and $0$
  otherwise, 

\begin{equation} 
\left\{
\begin{array}{lll}
NY & \to & YY \\
YN & \to & YY \\

N\sigma & \to & Y\sigma \\
\sigma N & \to & \sigma Y \\

Y\sigma' & \to & N\sigma' \\
\sigma' Y & \to & \sigma' N \\

\sigma\sigma' & \to & NY \\
\sigma'\sigma & \to & YN \\
 \end{array}
\right.
\end{equation}

taking 
\begin{itemize}
\item $\Sigma = \{\sigma,\sigma'\}, Q = \{\sigma,\sigma',Y,N\}$, 
\item $\omega(\sigma) = \omega(Y) = 1$,
\item $\omega(\sigma') = \omega(N) = 0$.
\end{itemize}

In this protocol, the states $Y$ and $N$ are ``neutral'' elements for
our predicate but they should be understood as \emph{Yes} and
\emph{No}. They are the ``answers'' to the question: are there more $0$s
than $1$s.

This protocol is made such that the numbers of $\sigma$ and $\sigma'$ are
preserved except when a $\sigma$ meets a $\sigma'$. In that latter case, the two
agents are deleted and transformed into a $Y$ and a $N$. 

If there are initially strictly more $\sigma$ than $\sigma'$, from the fairness
condition, each $\sigma'$ will be paired with a $\sigma$ and at some point no
$\sigma'$ will left. By fairness and since there is still at least a
$\sigma$, a configuration containing only $\sigma$ and $Y$s will be
reached. Since in such a configuration, no rule can modify the state
of any agent, and since the output is defined and equals to $1$ in
such a configuration, the protocol is correct in this case

By symmetry, one can show that the protocol outputs $0$ if there are
initially strictly more $\sigma'$ than $\sigma$. 

Suppose now that initially, there are exactly the same number of $\sigma$ and $\sigma'$. By fairness, there exists a step when no more
agents in the state $\sigma$ or $\sigma'$ left. Note that at the moment where the last
$\sigma$ is matched with the last $\sigma'$, a $Y$ is created. Since this $Y$ can
be ``broadcasted'' over the $N$s, in the final configuration all
agents are in the state $Y$ and thus the output is correct.

This protocol is Pavlovian, since it corresponds to the following
payoff matrix. 

\begin{center}
\begin{tabular}{llcccc} 
  &   & \multicolumn{4}{c}{Opponent}  \\
  &  &  \multicolumn{1}{c}{ ~\textsc{N}~~} & \textsc{Y} & \textsc{$\sigma$} & { \textsc{$\sigma'$} } \\ \cline{3-6}
 & \textsc{N} &  \multicolumn{1}{|c}{~$1$~~} & $-1$  & $-1$
  &\multicolumn{1}{c|}{$1$}\\

Player~~~ & \textsc{Y} &  \multicolumn{1}{|c}{~$0$~~} & $1$ & $1$
&\multicolumn{1}{c|}{$-1$}\\

&  \textsc{$\sigma$}   &  \multicolumn{1}{|c}{~$0$~~} & $0$ & $0$
&\multicolumn{1}{c|}{$-1$}\\ 

&  \textsc{$\sigma'$}   &  \multicolumn{1}{|c}{~$0$~~} & $0$ &$-1$
&\multicolumn{1}{c|}{$0$}\\

\cline{3-6}
\end{tabular}
\end{center}

\end{proof}

\section{Bounds on the Power of Pavlovian Population Protocols}
\label{sec:bound}

We proved that predicates $[x.\sigma=0]$, $[x.\sigma \ge 1]$, $[x.\sigma \ge 2]$ can be computed by some Pavlovian population protocols, as well as $[x.\sigma \ge x.\sigma']$.

It is clear that the subset of the predicates computable by
Pavlovian population protocols is closed by negation: just switch the
value of the individual output function of a protocol computing a
predicate to get a protocol computing its negation.

Notice that, unlike what happens for general population protocols, 
composing Pavlovian population protocols into a Pavlovian population
protocol is not easy. It is not clear whether Pavlovian
computable predicates are closed by conjunctions: classical
constructions for general population protocols can not be used directly.

The power of Pavlovian population protocols is actually rather limited as they can count up to $2$, but not $3$.

\begin{theorem}
There is no Pavlovian protocol that computes the threshold predicate $[x.\sigma \ge 3]$, which is true when there are at least $3$ occurrences of input symbol $\sigma$ in the input $x$.
\end{theorem}






\begin{proof}
  We will prove this by contradiction. Assume there exists such a Pavlovian protocol. Without loss of generality we may assume that $\Sigma=\{0,\sigma\}$ is a subset of the set of states $Q$.

  As the protocol is Pavlovian, and hence symmetric, any rule $q q \to q' q''$, is such that $q'=q''$, that is to say of the form $qq \to q'q'$ for all $q \in Q$.

  
Let then consider the sequence of rules such that $ \sigma \sigma \to q_1 q_1 \to q_2 q_2 \to \dots \to q_k q_k \to \dots $ where $\sigma , q_1 q_2,q_3,\dots,q_k\in Q$.  


Since $Q$ is finite, there exist two distinct integers $k$ and $\ell$ such that $q_k=q_\ell$ and $k<\ell$.  

The case $k+1=\ell$ is not possible. Indeed, we would have the rule $q_k q_k \to q_k q_k$. Consider the inputs $x_3$ and $x_4$ such that $x_3= \{\sigma,\sigma\}$ and $x_4= \{\sigma,\sigma,\sigma,\sigma\}$.  $x_4$ must be accepted. From $x_4$ there is a derivation $x_4\to \{q_1,q_1, \sigma ,\sigma\} \to \{q_1 ,q_1,q_1 q_1\} \to^* \{q_k,q_k,q_k,q_k\}$. This latter configuration is terminal from the above rule.  Since $x_4$ must be accepted, we must have $\omega(q_k)=1$. However, from $x_3$ there is a derivation $x_3\to \{q_1,q_1\} \to^* \{q_k,q_k\}$, where the last configuration is also terminal. We reach a contradiction, since its output would be $\omega(q_k)=1$, whereas $x_3$ must be rejected.

Hence, $k+1<\ell$, and $q_k q_k \to q_{k+1} q_{k+1} \to \dots \to q_\ell q_\ell \to q_{k} q_{k} $. Let $T$ be then the set of  states $T=\{q_i : k \leq i \leq \ell\}$.

Since $ q_i q_i \to q_{i+1} q_{i+1}$ is among the rules, since the protocol is Pavlovian with a matrix $M$, and by definition of  Pavlov behavior, we must have $q_{i+1} = BR_{\neq q_i}(q_i)$ (with the convention that $q_{\ell+1}$ is $q_k$).  So, $BR(q_i)$ can be $q_{i+1}$ or 
$q_i$.

Let then discuss the rules 
\begin{equation} \label{eq:eqq}
q_i q_j \to q'_i q'_j
\end{equation}
for $q_i,q_j \in T$.

There are three possibilities for the value of $q'_i$:
\begin{enumerate}
\item  $ q'_i = q_i$ if  $M_{q_i q_j}\geq \Delta$
\item $ q'_i = q_{j}$,  if  $M_{q_i q_j}\ <\Delta$ and if $ q_j = BR_{\neq q_i}(q_j)$
\item $ q'_i = q_{j+1}$  if  $M_{q_i q_j}\ <\Delta$ and if $ q_{j+1} = BR_{\neq q_i}(q_j)$
\end{enumerate}
In any case, we see that the value of $q'_i$ is in $T$. 

Symmetrically, we have three possibilities for $q'_j$, all of them in $T$.

Hence, all rules of the form \eqref{eq:eqq} preserve $T$: we have $q'_i,q'_j \in T$, as soon as $q_i, q_j \in T$.

Consider still then the inputs $x_3$ and $x_4$ such that $x_3= \{\sigma,\sigma\}$ and $x_4= \{\sigma,\sigma,\sigma,\sigma\}$.  From $x_4$ there is a derivation $x_4\to \{q_1,q_1, \sigma ,\sigma\} \to \{q_1 ,q_1,q_1 q_1\} \to^* \{q_k,q_k,q_k,q_k\}$. From this last configuration, by above remark, the state of all agents will be in $T$. As $x_4$ must be accepted, ultimately all agents will be in states that belong to $T$ whose image by $\omega$ is $1$. Consider now $x_3$. From $x_3$ there is a derivation $x_3\to \{q_1,q_1\} \to^* \{q_k,q_k\}$ that then will go trough all configurations $\{q_i q_i\}$, for the $q_i \in T$  in turn. This can not eventually stabilize to elements whose image by $\omega$ is $0$,  as some of the elements of $T$ have image $1$ by $\omega$, and hence $x_3$ is not accepted. This yields a contradiction, and hence such a Pavlovian protocol can not exist.
\end{proof}













\section{The Power of Symmetric Population Protocols}
\label{sec:fin}




Pavlovian Population protocols are symmetric. We just proved that they have a very limited computational power.  However, assuming population protocols symmetric (not-necessarily Pavlovian) is not truly a restriction.

\begin{proposition}
Any  population protocol can be simulated by a symmetric  population
protocol, as soon as the population is of size $\ge 3$.
\end{proposition}

Before proving this proposition, we state the (immediate) main consequence.

\begin{corollary}
A predicate is computable by a symmetric population protocol if and only
  if it is semilinear.
\end{corollary}

\begin{proof}
To a population protocol $(Q,\Sigma,\iota,\omega,\delta)$, with
$Q=\{q_1,\cdots,q_n\}$ associate population protocol
$(Q \cup Q',\Sigma,\iota,\omega,\delta')$ with
$Q'=\{q'_1,\cdots,q_n'\}$, $\omega(q')=\omega(q)$ for all $q \in Q$, and for all rules
$$qq \to \alpha \beta$$ in $\delta$, the following rules in $\delta'$:
$$
\left\{
\begin{array}{lll}
qq' &\to& \alpha\beta \\
q'q &\to& \beta\alpha \\
qq &\to& q'q' \\
q'q' &\to& qq \\
q \gamma &\to & q' \gamma \\
q' \gamma &\to& q \gamma \\
\gamma q & \to & \gamma q' \\
\gamma q' & \to & \gamma q \\
\end{array}
\right.
$$
for all $\gamma \in Q \cup Q', \gamma \neq q, \gamma \neq q'$,
and for all pairs of rules
$$
\left\{
\begin{array}{lll}
qr &\to& \alpha\beta \\
rq &\to& \delta\epsilon \\
\end{array}
\right.
$$
with $q,r \in Q$, the following rules in $\delta'$:
$$
\left\{
\begin{array}{lll}
qr' &\to& \alpha\beta \\
r'q & \to & \beta \alpha \\
rq' & \to & \delta \epsilon \\
q'r & \to & \epsilon \delta. \\
\end{array}
\right.
$$
The obtained population protocol is clearly symmetric. Now the first
set of rules guarantees that a state in $Q$ can always be converted to
its primed version in $Q'$ and vice-versa. By fairness, whenever a
rule $qq \to \alpha\beta$ (respectively $qr \to \alpha\beta$) can be applied, then the corresponding two first
rules of the first set of rules (resp. of the second set of rules) can
eventually be fired after possibly some conversions of states into their
primed version or vice-versa.
\end{proof}

\bibliographystyle{plain}

\end{document}